\documentclass{llncs}[11pt]

\usepackage{latexsym,amsfonts}
\usepackage{amsmath}
\usepackage{amssymb}
\usepackage{array}

\usepackage{graphicx,color}
\usepackage{latexsym,amsfonts}
\usepackage{amsmath}
\usepackage{amssymb}
\usepackage{bbm}

\usepackage{algorithmic}
\usepackage{algorithm}
\usepackage{tikz}
\usepackage{url}

\usepackage{newfloat}
\DeclareFloatingEnvironment[fileext=frm,placement={!ht},name=Frame]{myfloat}

\usetikzlibrary{calc}

\newenvironment{myalgorithm}[2][\columnwidth]
{\begin{myfloat}[!b]
	\begin{center}
   \begin{minipage}{#1}
     \begin{algorithm}[H]
     \caption{#2}
     \algsetup{linenosize=\normalsize, linenodelimiter=.}
       \begin{algorithmic}[1]\normalsize}
{      \end{algorithmic}
     \end{algorithm}
   \end{minipage}\vspace{1em}
 \end{center}
 \end{myfloat}}

\newcommand{\maxrank}{\mathsf{MaxRank}}
\newcommand{\rank}{\mathsf{rank}}
\newcommand{\proj}{\mathsf{proj}}
\newcommand{\vote}{\mathsf{vote}}

\newcommand{\capac}{\mathsf{cap}}
\newcommand{\wt}{\mathsf{wt}}
\newcommand{\res}{\mathsf{residual}}
\newtheorem{new-claim}{Claim}

\title{Popular Matchings with Multiple Partners}
\author{Florian Brandl\inst{1}\ \ \and \ \ Telikepalli Kavitha\inst{2}}
\institute{Technische Universit\"at M\"unchen, Germany. \email{brandlfl@in.tum.de} \and Tata Institute of Fundamental Research, India. \email{kavitha@tcs.tifr.res.in}}

\begin{document}
\pagestyle{plain}
\maketitle

\begin{abstract}
  Our input is a bipartite graph $G = (A \cup B,E)$ where each vertex in $A \cup B$ has a preference list strictly ranking its neighbors.
  The vertices in $A$ and in $B$ are called {\em students} and {\em courses}, respectively.
  Each student $a$ seeks to be matched to $\capac(a) \ge 1$ courses while each course $b$ seeks $\capac(b) \ge 1$ many students to be matched
  to it. The Gale-Shapley algorithm computes a pairwise-stable matching (one with no blocking edge) in $G$ in linear time. We consider the
  problem of computing a {\em popular} matching in $G$ -- a matching $M$ is popular if $M$ cannot lose an election to any matching where
  vertices cast votes for one matching versus another.
  Our main contribution is to show that a max-size popular matching in $G$ can be computed by the {\em 2-level Gale-Shapley} algorithm in
  linear time. This is an extension of the classical Gale-Shapley algorithm and we prove its correctness via linear programming.
\end{abstract}

\section{Introduction}
\label{intro}

We study the many-to-many matching problem in bipartite graphs: formally, this is given by a set $A$ of vertices (these vertices will be
called {\em students}) and a set $B$ of vertices (these will be called {\em courses}), where every vertex $u$ has an integral capacity 
$\capac(u) \ge 1$. Every student $a$ seeks to be matched to $\capac(a)$ many courses and every course $b$ seeks
$\capac(b)$ many students to be matched to it. Moreover, every student $a\in A$ has a strict ranking $\succ_a$ over courses that are acceptable
to $a$ and every course $b$ has a strict ranking $\succ_b$ over students that are acceptable to $b$. The set of mutually acceptable pairs
is given by $E\subseteq A\times B$. Thus our input is a bipartite graph $G = (A \cup B,E)$ and
the preferences of a vertex are expressed as an ordered list of its neighbors, e.g., $u\colon v, v'$ 
denotes the preference $v\succ_u v'$, i.e., $u$ prefers $v$ to $v'$.
 
\begin{definition}
  A matching $M$ in $G = (A\cup B, E)$ is a subset of $E$ such that $|M(u)| \le \capac(u)$ for each
  $u \in A \cup B$, where $M(u) = \{v: (u,v) \in M\}$.%
\end{definition}

The goal is to compute an {\em optimal} matching in $G$. The usual definition of optimality here
has been {\em pairwise-stability}~\cite{Roth84b}. A matching $M$ in $G$ is said to be pairwise-stable if there is no student-course 
pair $(a,b)$ that ``blocks'' $M$. We say a pair $(a,b)\in E\setminus M$ blocks $M$ if (1)~either $a$ has less than $\capac(a)$ partners in $M$ or 
$a$ prefers $b$ to its worst partner in $M$ and (2)~either $b$ has less than $\capac(b)$ partners in $M$ or $b$ prefers $a$ to its 
worst partner in $M$. It is known that pairwise-stable matchings always exist~\cite{Roth84b} and the Gale-Shapley algorithm~\cite{GS62}
for the one-to-one variant, or the {\em marriage} problem, can be easily generalized to find such a matching in $G = (A \cup B,E)$.
The many-to-one variant, also called the {\em hospitals/residents} problem, where $\capac(a) = 1$ for every $a \in A$,
was studied by Gale and Shapley~\cite{GS62} who showed that their algorithm for the marriage problem extends to the hospitals/residents problem.

Since a (pairwise) stable matching is a maximal matching in $G$, its size is at least $|M_{\max}|/2$, where 
$M_{\max}$ is a max-size matching in $G$. This bound can be tight as shown by the following simple 
example: let $A = \{a,a'\}$ and $B=\{b,b'\}$ where each vertex has capacity 1 and the edge set is 
$E = \{(a,b),(a,b'),$ $(a',b)\}$. The preferences are shown in the table below.
Here the only stable matching (red line) is $S = \{(a,b)\}$, which is of size $1$.
However, the max-size matching (dashed lines) $M_{\max} = \{(a',b),(a,b')\}$ is of size $2$.

	\begin{center}
		
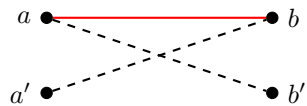
		\begin{minipage}[c]{0.45\textwidth}
			
			\centering
			\begin{align*}
				&a\colon b, b'\qquad\qquad &&b\colon a, a'\\
				&a'\colon b &&b'\colon a\\
			\end{align*}
			
		\end{minipage}
		\hfill
		\begin{minipage}[c]{0.45\textwidth}
			
			\centering 
  			\begin{tikzpicture}
		
  				\def\radius{.15cm}
		
  				\node[label = left:{$a$}, draw, circle, fill, minimum size = \radius, inner sep = 0pt] (a) at (0,.5) {};
  				\node[label = left:{$a'$}, draw, circle, fill, minimum size = \radius, inner sep = 0pt] (a') at (0,-.5) {};

  				\node[label = right:{$b$}, draw, circle, fill, minimum size = \radius, inner sep = 0pt] (b) at (3,.5) {};
  				\node[label = right:{$b'$}, draw, circle, fill, minimum size = \radius, inner sep = 0pt] (b') at (3,-.5) {};

  				\draw[red, thick] (a) to (b);
  				\draw[dashed, thick] (a) to (b');
  				\draw[dashed, thick] (a') to (b);		
	
  			\end{tikzpicture}
	
		\end{minipage}
	
	\end{center}

It can be shown that all pairwise-stable matchings have to match the same set of vertices and every 
vertex gets matched to the same capacity in every pairwise-stable matching. In the hospitals/residents setting,
this is popularly called the ``Rural Hospitals Theorem''~\cite{GS85,Roth86a}. More precisely, 
Roth~\cite{Roth86a} showed that not only is every hospital matched to the same number of residents in 
every stable matching, but moreover, every hospital that is not matched up to its capacity in some stable 
matching is actually matched to the same \emph{set of residents} in any stable matching. Thus the 
notion of stability is very restrictive.

From a social point of view, it seems desirable to have a high number of students registered for courses to make effective use of available
resources. Similarly, in the hospitals/residents setting, it seems desirable to have a higher number of residents matched to hospitals in order
to keep few residents unemployed and guarantee sufficient staffing for hospitals. The latter point particularly applies to 
rural hospitals that oftentimes face the problem of being understaffed with residents by the National Resident Matching 
Program in the USA (cf. \cite{Roth84,Roth86a}). 
This motivates relaxing the notion of ``absence of blocking edges'' to a weaker notion of stability so as to obtain matchings 
that are guaranteed to be significantly larger than $|M_{\max}|/2$. Note that we do not wish to ignore the preferences of 
vertices and impose a max-size matching on them as such a way of assignment may be highly undesirable from a social viewpoint. 
Instead our approach is to replace the \emph{local} stability notion of ``no blocking edges'' with a weaker notion of {\em global} stability
that achieves more ``global good'' in the sense that its size is always at least $\gamma\cdot|M_{\max}|$ for some $\gamma > 1/2$.

\subsection{Popularity}
To this end, we consider the notion of popularity, which was introduced by G\"ardenfors~\cite{Gar75} in the one-to-one matching setting
or the stable marriage problem: the input here consists
of a set of men and a set of women, where each person seeks to get matched to at most one person from the opposite sex. Thus
every vertex has capacity $1$ here.
Popularity is based on voting by vertices on the set of feasible matchings.
In the one-to-one setting, the preferences of a vertex over its neighbors are extended to preferences over matchings by postulating that every
vertex orders matchings in the order of its partners in these matchings. 
This postulates that vertices do not care which other pairs are formed.
A matching is popular if it never loses a head-to-head election against any matching where each vertex casts a vote.
Thus popular matchings are (weak) {\em Condorcet winners}~\cite{condorcet} in the corresponding voting instance.
The Condorcet paradox shows that collective preferences can be cyclic and so there need not be a Condorcet winner;
this is the source of many impossibility results in social choice theory such as Arrow's impossibility theorem.

However, in the context of matchings in the stable marriage problem, G\"ardenfors~\cite{Gar75} showed that every stable matching is popular.
Hence the fact that stable matchings always exist here~\cite{GS62} implies that popular matchings always exist.
This is quite remarkable given that popular matchings correspond to (weak) Condorcet winners.
In the one-to-one matching setting, there is a vast literature on popular matchings~\cite{BIM10,HK11,Kav12,Hirakawa-MatchUp15,CK16,Kav16,HK17}.

Here we generalize the notion of popularity to the many-to-many matching setting. This requires us to specify how vertices vote over different
subsets of their neighbors. In particular, one may want to allow a single vertex to cast multiple votes if its capacity is greater than~$1$.
Our definition of voting by a vertex between two subsets of its neighbors is the following: first remove all vertices that are contained in both
sets; then find a bijection from the first set to the second set and compare every vertex with its image under this bijection (if the sets are
not of equal size, we add dummy vertices that are less preferred than all non-dummy vertices); the number of wins minus the number of losses
is cast as the vote of the vertex. The vote may depend on the bijection that is chosen, however. 

Our definition is based on the bijection that \emph{minimizes} the vote, which results in a rather restrictive notion of popularity. 
We show however that even for this notion of popularity, every stable matching is popular.
In particular, popular matchings always exist. As a consequence, popular matchings always exist for every notion of popularity that is less
restrictive than our notion of popularity. Our goal is to find a {\em max-size} popular matching and crucially, it turns out that the size of a
max-size popular matching is independent of the bijection that is
chosen for the definition of popularity. We formalize all these notions below.

\subsection{Definitions}
In the one-to-one setting, given any two matchings $M_0, M_1$
and a vertex $u$, we say $u$ prefers $M_0$ to $M_1$ if $u$ prefers $M_0(u)$ 
to $M_1(u)$, where $M_i(u)$ is $u$'s partner in $M_i$, for $i = 0,1$, and  
we say ``$M_i(u) =$ null'' if $u$ is left unmatched in matching $M_i$
-- note that the null option is the least preferred state for any vertex.
Define the function $\vote_u(v,v')$ for any vertex $u$ and neighbors $v, v'$ of $u$
as follows: $\vote_u(v,v')$ is 1 if $u$ prefers $v$ to $v'$, it is 
$-1$ if $u$ prefers $v'$ to $v$, and it is 0 otherwise (i.e., if $v = v'$). In the 
one-to-one setting, $\Delta_u(M_0,M_1)$, which is $u$'s vote for $M_0$ versus $M_1$, is defined to
be $\vote_u(M_0(u),M_1(u))$.

In the many-to-many setting, while comparing one matching with another, we allow a vertex to cast more than one vote.
For instance, when we compare the preference of vertex $u$ with $\capac(u) = 3$ for
$S_0 = \{v_1,v_2,v_3\}$ versus $S_1 = \{v_4,v_5,v_6\}$ (where $v_1 \succ_u v_2 \succ_u \cdots \succ_u v_6$), 
we would like $u$'s vote to capture the fact that $u$ is better-off by {\em 3 partners} in $S_0$ 
when compared to $S_1$. 
So we define $u$'s vote for $S_0$ versus $S_1$ as follows.
Let $S_0, S_1$ be any two subsets of the set of $u$'s neighbors where we add some occurrences of ``null'' to the smaller of
$S_0, S_1$ to make the two sets of the same size. We will view the sets $S'_0 = S_0\setminus S_1$ and $S'_1 = S_1\setminus S_0$
as arrays $\langle S'_i[1],\ldots,S'_i[k]\rangle$ (for $i = 0,1$) where $k = |S_0| - |S_0\cap S_1| = |S_1| - |S_0\cap S_1|$.
The preference of vertex $u$ for $S_0$ versus $S_1$, denoted by $\delta_u(S_0,S_1)$, is defined
as follows:
\begin{equation}
  \label{eq:delta}
  \delta_u(S_0,S_1) = \min_{\sigma \in \Pi[k]}\sum_{i = 1}^k \vote_u(S'_0[i],S'_1[\sigma(i)]),
\end{equation}  
where $\Pi[k]$ is the set of permutations on $\{1,\ldots,k\}$.
Let $\Delta_u(M_0,M_1) = \delta_u(S_0,S_1)$, where $S_0 = M_0(u)$ and $S_1 = M_1(u)$.
So $\Delta_u(M_0,M_1)$ counts the number of votes by $u$ for $M_0(u)$ versus $M_1(u)$ 
when the sets $S'_0 = M_0(u)\setminus M_1(u)$ and $S'_1 = M_1(u)\setminus M_0(u)$ are being compared 
in the order that is most adversarial or {\em negative} for $M_0$. That is, this order 
$\sigma \in \Pi[k]$ of comparison between elements of $S'_0$ and $S'_1$ gives the least value for 
$n^+ - n^-$,  where $n^+$ is the number of indices $i$ such that $S'_0[i] \succ_u S'_1[\sigma(i)]$ 
and $n^-$ is the number of indices $i$ such that $S'_0[i] \prec_u S'_1[\sigma(i)]$.
Note that $\Delta_u(M_0,M_1) + \Delta_u(M_1,M_0) \le 0$ and it can be {\em strictly negative}.

For instance, when a vertex $u$ with $\capac(u) = 3$ compares two subsets $S_0 = \{v_1,v_3,v_5\}$ and 
$S_1 = \{v_2,v_4,v_6\}$ (where $v_1 \succ_u v_2 \succ_u \cdots \succ_u v_6$), we
have $\delta_u(S_0,S_1) = -1$ since comparing the following pairs results in the least value of
$\delta_u(S_0,S_1)$: this pairing is ($v_1$ with $v_6$),  ($v_3$ with $v_2$), ($v_5$ with $v_4$). This 
makes $\delta_u(S_0,S_1) = 1 -1 -1 = -1$. While computing $\delta_u(S_1,S_0)$, the pairing would be 
($v_2$ with $v_1$),  ($v_4$ with $v_3$), ($v_6$ with $v_5$): then $\delta_u(S_1,S_0) = -1 -1 -1 = -3$.

 For any two matchings $M_0$ and $M_1$ in $G$, we compare them using the
 function $\Delta(M_0,M_1)$ defined as follows. 
\begin{equation}
\label{def:Delta}
\Delta(M_0,M_1) = \sum_{u \in A \cup B}\ \Delta_u(M_0,M_1).
\end{equation}

We say $M_0$ is at least as popular as $M_1$ if $\Delta(M_0,M_1) \ge 0$ and $M_0$ is more 
popular than $M_1$ if $\Delta(M_0,M_1) > 0$. If $\Delta(M_0,M_1) \ge 0$ then for every vertex $u$ 
in $A \cup B$: {\em no matter in which order} the elements of $S'_0 = M_0(u)\setminus M_1(u)$ and 
$S'_1 = M_1(u)\setminus M_0(u)$ are compared against each other by $u$ in the evaluation of 
$\Delta_u(M_0,M_1)$ -- when we sum up the total number of votes cast by all vertices,
the votes for $M_1$ can {\em never} outnumber the votes for $M_0$.

\begin{definition}
\label{def:pop-match}  
$M_0$ is a popular matching in $G = (A\cup B,E)$ if $\Delta(M_0,M_1) \ge 0$ for all matchings $M_1$ 
in $G$.
\end{definition}

Thus for a matching $M_0$ to be popular, it means that $M_0$ is at least as popular as every 
matching in $G$, i.e., there is no matching $M_1$ with $\Delta(M_0,M_1) < 0$. If there 
exists a matching $M_1$ such that $\Delta(M_0,M_1) < 0$ then this is taken as a certificate of 
{\em unpopularity} of $M_0$. Note that it is possible that 
both $\Delta(M_0,M_1)$ and $\Delta(M_1,M_0)$ are negative, i.e., for each vertex $u$ there is some 
order of comparison between the elements of $S'_0 = M_0(u)\setminus M_1(u)$ with those in 
$S'_1 = M_1(u)\setminus M_0(u)$ so that when we sum up the total number of votes cast by all the vertices, 
the number for $M_1$ is more than the number for $M_0$; similarly for each $u$ there is another order of 
comparison between the elements of $S'_0$ with those in $S'_1$ so that when we sum up the total number of 
votes cast by all the vertices, the number for $M_0$ is more than the number for $M_1$. In this case neither
$M_0$ nor $M_1$ is popular. It is not obvious whether popular matchings always exist in $G$.

Our definition of popularity may seem too strict and restrictive since for each vertex $u$, we 
choose the most negative or adversarial ordering for $M_0(u)\setminus M_1(u)$ versus $M_1(u)\setminus M_0(u)$
while calculating $\Delta_u(M_0,M_1)$. A more relaxed definition may be to order the sets 
$S'_0 = M_0(u)\setminus M_1(u)$ and  $S'_1 = M_1(u)\setminus M_0(u)$ in increasing order of 
preference of $u$ and take $\sum_i \vote_u(S'_0[i],\,S'_1[i])$ as $u$'s vote. 
An even more relaxed definition may be to choose 
the most favorable or {\em positive} ordering for $S'_0$ versus $S'_1$ while calculating 
$\Delta_u(M_0,M_1)$. 
Note that as per (\ref{eq:delta}) we have:
\begin{equation}
  \label{eqn:weakly-popular}
  -\Delta_u(M_0,M_1) = -\min_{\sigma \in \Pi[k]}\sum_{i = 1}^k \vote_u(S'_0[i],S'_1[\sigma(i)]) = \max_{\pi \in \Pi[k]}\sum_{i = 1}^k \vote_u(S'_1[i],S'_0[\pi(i)]).
\end{equation}

\begin{definition}
\label{def:weakly-pop}
Call a matching $M_1$ {\em weakly popular} if $\Delta(M_0,M_1) \le 0$, i.e., $-\Delta(M_0,M_1) \ge 0$, 
for all matchings $M_0$ in $G$.
\end{definition}

Thus it follows from (\ref{def:weakly-pop}) that $M_1$ is a weakly popular matching if the sum of votes for $M_1$
is at least the sum
of votes for any matching $M_0$ when each vertex $u$ compares $M_1(u)\setminus M_0(u)$ versus
$M_0(u)\setminus M_1(u)$ in the ordering (as given by $\pi$) that is most favorable for $M_1$. 
In the one-to-one setting, we have $\Delta(M_0,M_1) + \Delta(M_1,M_0) = 0$ for any pair of matchings $M_0, M_1$
since $\Delta_u(M_0,M_1) = \vote_u(M_0(u),M_1(u)) = -\vote_u(M_1(u),M_0(u)) = -\Delta_u(M_1,M_0)$ for each $u$;
thus the notions of ``popularity'' and ``weak popularity'' coincide here.
In the many-to-many setting, weak popularity is a more relaxed notion than popularity.

We choose a strong definition of popularity so that a matching that is popular according to our notion
will also be popular according to any notion ``in between'' between popularity and weak popularity. However
this breadth may come at a price as it could be the case that a max-size weakly popular matching is larger
than a  max-size popular matching.

\subsection{Our results}
We will show that every pairwise-stable matching in $G = (A \cup B,E)$ is popular, 
thus our (seemingly strong) definition of popularity is a relaxation of pairwise-stability.
We will present a simple linear time algorithm for computing a max-size popular matching $M_0$
in $G$ and show that $|M_0| \ge \frac{2}{3}\cdot|M_{\max}|$. 

We also show that $M_0$ is more popular than every {\em larger} matching, i.e., $\Delta(M_0,M_1) > 0$ (refer to (\ref{def:Delta}))
for any matching $M_1$ that is larger than $M_0$. Thus $M_0$ is also a {\em max-size weakly popular matching} in $G$
as no matching $M_1$ larger than $M_0$ can be weakly popular due to the fact that $\Delta(M_0,M_1) > 0$.
Thus surprisingly, we lose nothing in terms of the size of our matching by sticking to a strong notion of popularity.

Akin to the rural hospitals theorem, we show that all max-size popular matchings have to match the 
same set of vertices and every vertex gets matched to the same capacity in every max-size popular 
matching. However, even in the hospitals/residents setting, hospitals that are not matched up to their capacity in some max-size popular matching do {\em not} need to be matched to the same sets of residents in any 
max-size popular matching, which is in contrast to stable matchings \cite{Roth86a}.

\subsubsection{Techniques.} 
Our algorithm is an extension of the 2-level Gale-Shapley algorithm from~\cite{Kav12} to find a 
max-size popular matching in a stable marriage instance. 
While the analysis of the 2-level Gale-Shapley algorithm in \cite{Kav12}
is based on a structural characterization of popular matchings (from \cite{HK11}) on forbidden 
alternating paths and alternating cycles, we use linear programming here to show a simple and self-contained proof
of correctness of our algorithm. We would like to remark that the structural characterization from \cite{HK11}
and the proof of correctness from \cite{Kav12} 
can be extended (in a rather laborious manner) to show the correctness of the generalized algorithm in
our more general setting as well, however our proof of correctness is much simpler and this yields a much easier
proof of correctness of the algorithm in \cite{Kav12}.
Our linear programming techniques are based on a linear program used in \cite{KMN09} to find a popular fractional 
matching in a bipartite graph with {\em 1-sided  preference lists}.

\subsection{Background and related work}
The first algorithmic question studied for popular matchings was in the domain of 1-sided preference
lists~\cite{AIKM07} where it is only vertices on the left, who are {\em agents}, that have preferences; the
vertices on the right are {\em objects} and they have no preferences. Popular matchings need not always exist 
here, however fractional matchings that are popular always exist and can be computed in polynomial time
via linear programming~\cite{KMN09}. 
Popular matchings always exist in any instance of the stable marriage problem  with strict preference
lists since every stable matching is popular~\cite{Gar75}.

Efficient algorithms to find a max-size popular matching in
a stable marriage instance are known~\cite{HK11,Kav12} and a subclass of max-size  popular matchings called
dominant matchings was studied in \cite{CK16} and it was shown that these matchings coincide with stable matchings
in a larger graph.
A polynomial time algorithm was shown in \cite{Kav16} to find a min-cost popular half-integral matching
when there is a cost function on the edge set and it was shown in~\cite{HK17} that the popular fractional matching polytope here
is half-integral.
When preference lists admit ties, the problem of determining if a stable marriage instance $(A \cup B,E)$ admits a popular matching
or not is NP-hard~\cite{BIM10} and the NP-hardness of this problem holds even when ties are allowed on only one side
(say, in the preference lists of vertices in $A$)~\cite{CHK15}. 

Very recently and independent of our work, the problem of computing a max-size popular matching in an extension of the
hospitals/residents problem, i.e., in the {\em many-to-one} setting, was considered by Nasre and Rawat~\cite{NR17}. The notion of 
``more popular than'' in \cite{NR17} is weaker than ours: in order to compare matchings $M_0$ and $M_1$, in \cite{NR17} every 
hospital $h$ orders $S'_0 = M_0(h)\setminus M_1(h)$ and $S'_1 = M_1(h)\setminus M_0(h)$ in increasing order of preference of $h$ 
and $\sum_i \vote_h(S'_0[i],\,S'_1[i])$ is $h$'s vote for $M_0$ versus $M_1$. An efficient algorithm was shown for their problem
by reducing it to a stable matching problem in a larger graph -- this closely follows the method and techniques 
in \cite{HK11,Kav12,CK16} for the max-size popular matching problem in the one-to-one setting. Note that popularity as per their definition 
is ``in between'' our notions of popularity and weak popularity.

The stable matching problem in a marriage instance has been extensively studied -- we refer to the books~\cite{GI89,manlove2013}
on this topic. The problem of computing stable matchings or its variants in the hospitals/residents 
setting is also well-studied~\cite{huang2010,askalidis2013,HIM16,IMS00,IMS03}.
The stable matching algorithm in the 
hospitals/residents problem has several real-world applications -- it is used to match residents to hospitals in Canada~\cite{CRMS}
and in the USA~\cite{NRMP}. The many-to-many stable matching problem has also received considerable attention~\cite{Roth84b,Bla88,Sot99}.

\section{Our algorithm}\label{sec:hosp-res-algo}

A first attempt to solve the max-size popular matching problem in a many-to-many instance $G = (A \cup B, E)$ may be to define an
equivalent one-to-one instance $G' = (A' \cup B', E')$ by making $\capac(u)$ copies of each $u \in A \cup B$ and $\capac(a)\cdot\capac(b)$
many copies of each edge $(a,b)$; the max-size popular matching problem in $G'$
can be solved using the algorithm in \cite{Kav12} and the obtained matching $\tilde{M}$ in $G'$ can be mapped to a matching $M$
in $G$. In the first place, one should ensure that there are no multi-edges in $M$. 
The matching $M$ will be popular,
however it is not obvious that $M$ is a {\em max-size} popular matching in $G$ as
every popular matching in $G$ need not be realized as some popular matching in $G'$ (Appendix~A has such an example).
Thus one needs to show that there is at least one max-size popular matching in $G$ that can be
realized as a popular matching in $G'$; we do not pursue this approach here as the running time of the max-size popular matching
algorithm in $G'$ is $O(mn)$ (linear in the size of $G'$) where $|E| = m$ and $|A|+|B|=n$.

In this section we describe a simple $O(m+n)$ algorithm called the generalized {\em 2-level Gale-Shapley algorithm} 
to compute a max-size popular matching in $G = (A\cup B,E)$. This algorithm works on the graph 
$H = (A''\cup B, E'')$ defined as follows: $A''$ consists of two copies $a^0$ and $a^1$ of 
every student $a$ in $A$, i.e., $A'' = \{a^0, a^1: a \in A\}$. The set $B$ of
courses in $H$ is the same as in $G$ and the edge set here is $E'' = \{(a^0,b), (a^1,b): (a,b) \in E\}$.

\begin{myalgorithm}[14cm]{\ \ \ \em Input: $H = (A''\cup B,E'')$; \ \ \ \ Output: A matching $M$ in $H$}\label{alg:many-many}
\STATE{Initialize $Q = \{a^0: a\in A\}$ and $M = \emptyset$. Set $\res(a) = \capac(a)$ for all $a \in A$.}
\WHILE{$Q \ne \emptyset$}
\STATE{delete the first vertex from $Q$: call this vertex $a^i$.}
\WHILE{$a^i$ has one or more neighbors in $H$ to propose to \underline{and} $\res(a) > 0$}
\STATE{-- let $b$ be the most preferred neighbor of $a^i$ in $H$ that $a^i$ has not yet proposed to.} 

\COMMENT{{\em So every neighbor of $a^i$ in the current graph $H$ that is ranked better than $b$ is already matched to $a^i$ in $M$.}}
\STATE{-- add the edge $(a^i,b)$ to $M$.} 
\IF{$i=1$ and $b$ is already matched to $a^0$}
\STATE{-- delete the edge $(a^0,b)$ from $M$.}
\COMMENT{{\em So $(a^0,b)$ in $M$ gets replaced by $(a^1,b)$.}}
\ELSE
\STATE{-- set $\res(a) = \res(a) - 1$.}
\COMMENT{{\em since $|M(a)|$ has increased by 1}}

\IF{$b$ is matched to more than $\capac(b)$ partners in $M$}
\STATE{-- let $v^j$ be $b$'s worst partner in $M$. Delete the edge $(v^j,b)$ from $M$.}

\COMMENT{{\em Note that ``worst'' is as per preferences in $H$.}}
\STATE{-- set $\res(v) = \res(v) + 1$ and if $v^j \notin Q$ then add $v^j$ to $Q$.}
\ENDIF
\ENDIF
\IF{$b$ is matched to $\capac(b)$ many partners in $M$}
\STATE{-- delete all edges $(u,b)$ from $H$ where $u$ is a neighbor of $b$ in $H$ that is ranked worse than $b$'s worst partner in $M$.}
\COMMENT{{\em ``Worse'' is as per preferences in $H$.}}
\ENDIF
\ENDWHILE
\IF{$\res(a) > 0$ and $i =0$}
\STATE{-- add $a^1$ to $Q$.}\ \ \ \COMMENT{{\em Though $\res(a) > 0$, the condition in the above while-loop does not hold, i.e., $a^0$ has no neighbors in $H$ to propose to; hence $a^1$ gets activated.}}
\ENDIF
\ENDWHILE
\STATE{Return the matching $M$.}
\end{myalgorithm}

The preference list of $a^i$ (for $i = 0,1$) is exactly the same as the preference list of $a$.
The elements in the set $\{a^i: a \in A\}$ will be called {\em level~$i$} students, for $i = 0,1$.
Every $b \in B$ prefers any level~1 neighbor to any level~0 neighbor: within the set of level~$i$ 
neighbors (for $i= 0,1$), $b$'s preference order is the same as its original preference order. For 
instance, if a course $b$ has only 2 neighbors $a$ and $v$ in $G$ where $a \succ_b v$, the 
preference order of $b$ in $G'$ is: $a^1, \ v^1, \ a^0, \ v^0$.
The {\em sum} of capacities of $a^0$ and $a^1$ will be $\capac(a)$ and we will use $\res(a)$ to denote
the $\capac(a) - |M(a)|$, where $M$ is the current matching.
At any point in time, only one of $a^0$ and $a^1$ will be {\em active}  in our algorithm.

A description of our algorithm is given as Algorithm~\ref{alg:many-many}.
To begin with, all level~0 students are active in our algorithm and all level~1 students are 
inactive. We keep a queue $Q$ of all the active students and they propose as follows:
\begin{itemize}
\item every active student $a^i$, where $a$ is not fully matched, 
  proposes to its most preferred neighbor in $H$ that it has not yet proposed to (lines~4-5 of Algorithm~\ref{alg:many-many})
\item if $a^0$ has already proposed to all its neighbors in $H$ and $a$ is not fully matched,
then $a^0$ becomes inactive and $a^1$ becomes active and it joins the queue $Q$ (lines~20-21).
\end{itemize}

When a course $b$ receives a proposal from $a^i$, the vertex $b$ accepts this offer (in line~6).
In case $b$ is already matched to $a^0$ and it now received a proposal from $a^1$, the edge $(a^0,b)$ in $M$
is replaced by the edge $(a^1,b)$ (otherwise $b$ would end up being matched to $a$ with multiplicity 2
which is not allowed) -- this is done in lines 7-8 of Algorithm~\ref{alg:many-many}.

If $b$ is now matched to more than $\capac(b)$ partners then $b$ rejects its worst partner $v^j$ in the 
current matching and so $\res(v)$ increases by 1 and $v^j$
joins $Q$ if it is not already in $Q$ (in lines 11-13).
If $b$ is now matched to $\capac(b)$ partners then we delete all edges $(u,b)$ from $H$ where 
$u$ is a neighbor of $b$ in $H$ that is ranked worse than $b$'s worst partner in the current matching 
-- so no such resident $u$ can propose to $b$ later on in the algorithm (lines~16-17).
Once $Q$ becomes empty, the algorithm terminates. 

\smallskip

\noindent{\bf The matching $M_0$.}
Let $M$ be the matching returned by this algorithm and let $M_0$ be the matching in $G$ that is
obtained by projecting $M$ to the edge set of $G$,
i.e., $(a,b)\in M_0$ if and only if $(a^i,b)\in M$ for some $i \in \{0,1\}$.
We will prove that $M_0$ is a max-size popular matching in Section~\ref{sec:analysis}.

\section{The correctness of our algorithm}
\label{sec:analysis}
In this section we show a sufficient condition for a matching $N$ in $G$ to be popular. This is shown
via a graph called $G'_N$: this is a bipartite graph constructed using $N$ such that $N$ gets mapped to
a {\em simple matching} $N'$ in $G'_N$, i.e., $|N'(v)| \le 1$ for all vertices $v$ in $G'_N$.

The vertex set of $G'_N$ includes $A' \cup B'$ where $A' = \{a_i: a \in A\ \text{and} \ 1 \le i \le \capac(a)\}$ and
$B' = \{b_j: b \in B\ \text{and} \ 1 \le j \le \capac(b)\}$. That is, for each vertex $u \in A \cup B$, there are $\capac(u)$ many 
copies of $u$ in $G'_N$. For each edge $(a,b)$ in $G$ such that $(a,b) \in N$, we will arbitrarily choose a distinct 
$i \in \{1,\ldots,\capac(a)\}$ and a distinct $j \in \{1,\ldots,\capac(b)\}$ and include $(a_i,b_j)$ in $N'$. If $u \in A \cup B$ 
was not fully matched in $N$, i.e., it has less than $\capac(u)$ many partners in $N$, then some $u_k$'s will be left unmatched in 
$N'$.

\begin{itemize}
  \item[1.] For each edge $(a,b)$ in $G$ such that $(a,b) \notin N$, we will have edges $(a_i,b_j)$ in $G'_N$, for all 
$1 \le i \le \capac(a)$ and $1 \le j \le \capac(b)$.
  \item[2.] For each edge $(a,b) \in N$, we will have the edge $(a_i,b_j)$ in $G'_N$ where  $(a_i,b_j)\in N'$.
\end{itemize}

Thus for any edge $e = (a,b) \notin N$, there are $\capac(a)\cdot\capac(b)$ many copies of $e$ in $G'$: these are $(a_i,b_j)$ for
all $(i,j) \in \{1,\ldots,\capac(a)\}\times\{1,\ldots,\capac(b)\}$. However for any edge $(a,b) \in N$, there is only {\em one} edge 
$(a_i,b_j)$ in $G'_N$ where $(a_i,b_j)\in N'$,
in other words, the student $a_i$ is not adjacent in $G'_N$ to course $b_{j'}$ for $j' \ne j$ and similarly, the course
$b_j$ is not adjacent in $G'_N$ to student $a_{i'}$ for $i' \ne i$. Appendix~A has an example of $G'_N$ corresponding to
a matching $N$ in a many-to-one instance $G$ (see Fig.~\ref{fig:Appendix-example}).

There are also some new vertices called ``last resort neighbors'' in $G'_N$: for any $v \in A'\cup B'$,
there is one vertex $\ell(v)$ and every vertex $v$ ranks $\ell(v)$ at the bottom of its 
preference list. 
\begin{itemize}
  \item[3.] The edge set of $G'_N$ also contains the edges $(v,\ell(v))$ for each $v \in A' \cup B'$.
\end{itemize}

The purpose of the vertex $\ell(v)$ is to capture the state of $v\in A'\cup B'$ being left unmatched in any
matching so that every matching in $G$ gets mapped to an {\em $(A' \cup B')$-complete matching} in $G'_N$, i.e.,
one that matches all vertices in $A' \cup B'$.
We will use these last resort neighbors to obtain an $(A' \cup B')$-complete matching $N^*$ from $N'$.
\[N^* = N' \cup \{(v,\ell(v)): v\in A' \cup B'\ \text{and}\ v \ \text{is\ unmatched\ in}\ N'\}.\]
Thus if a vertex $u \in A\cup B$ was not fully matched in $N$,
then some $u_i$'s will be matched to their last resort neighbors in $N^*$. We now define edge weights in $G'_N$.

\begin{itemize}
\item For any edge $e = (a_i,b_j) \in A' \times B'$: the weight of edge $e$ is 
$\wt_N(e) = \vote_a(b,N^*(a_i)) + \vote_b(a,N^*(b_j))$, where $N^*(u_k)$ is $u_k$'s partner 
  in the $(A'\cup B')$-complete matching $N^*$. Thus $\wt_N(a_i,b_j)$ is the sum of votes of $a$ and $b$ for
  each other versus $N^*(a_i)$ and $N^*(b_j)$, respectively. We have $\wt_N(e) \in \{\pm 2, 0\}$ and 
$\wt_N(e) = 2$ if and only if $e$ blocks $N$.
\item For any edge $e = (v,\ell(v))$ where $v \in A' \cup B'$: the weight of edge $e$ is 
$\wt_N(e) = \vote_v(\ell(v),N^*(v))$. 
  Thus $\wt_N(v,\ell(v)) = -1$ if $v$ was matched in $N'$ 
  and $\wt_N(v,\ell(v)) = 0$ otherwise (in which case $N^*(v) = \ell(v)$).
\end{itemize}  

Observe that every edge $e \in N^*$ satisfies $\wt_N(e) = 0$. Thus the weight of the matching $N^*$ in $G'_N$ is 0.
Theorem~\ref{main-thm} below states that if {\em every} $(A' \cup B')$-complete matching in the graph
$G'_N$ has weight at most 0, then $N$ is a popular matching in $G$.

\subsection{A sufficient condition for popularity}

\begin{theorem}
  \label{main-thm}
  Let $N$ be a matching in $G$ such that every $(A' \cup B')$-complete matching in $G'_N$ has weight at most 0. Then $N$ is  popular.
\end{theorem}  
\begin{proof}
  For any matching $T$ in $G$, we will show a realization $T^*$ of $T$ in $G'_N$ such that $T^*$ is an $(A' \cup B')$-complete 
  matching and $\wt_N(T^*) = -\Delta(N,T)$. Thus if every $(A' \cup B')$-complete matching in $G'_N$ has weight at most 0, then
  $\wt_N(T^*) \le 0$, in other words, $\Delta(N,T) \ge 0$. Since $\Delta(N,T) \ge 0$ for all matchings $T$ in $G$, the matching $N$ will be
  popular. 

  In order to construct $T^*$, corresponding to each edge $(a,b) \in T$, we will find appropriate
  indices $s \in \{1,\ldots,\capac(a)\}$ and $t \in \{1,\ldots,\capac(b)\}$, where $(a_s,b_t)$ is in $G'_N$, such that 
  $(a_s,b_t) \in T^*$; there may also be some $(u_k,\ell(u_k))$ edges in $T^*$.
  \begin{itemize}
     \item[(i)] For every edge $(a,b) \in N \cap T$ do: if $(a_i,b_j) \in N^*$ then $(a_i,b_j)$ belongs to $T^*$ as well.
          \item[(ii)] For every $(a,b) \in T\setminus N$, we have to decide the indices $(s,t)$ such that $(a_s,b_t) \in T^*$.
       In the evaluation of $\Delta_a(N,T)$, while comparing the sets $N(a) \setminus T(a)$ and $T(a) \setminus N(a)$ (refer to Equation~(\ref{eq:delta}) in Section~\ref{intro}):
       \begin{itemize}
       \item[--]  let $b'$ be the course that $a$ compares $b$ with. 
So the matching $N^*$ contains the edge $(a_i,b'_j)$ for some $(i,j)$ and we now need to decide the index $k$ such that 
$T^*$ will contain $(a_i,b_k)$. In the evaluation of $\Delta_b(N,T)$, while comparing the sets $N(b) \setminus T(b)$ and $T(b) \setminus N(b)$:
       \begin{itemize}
           \item let $a'$ be the student that $b$ compares $a$ with. So the matching $N^*$ contains the edge $(a'_{i'},b_{j'})$ for some 
$(i',j')$. We include the edge $(a_i,b_{j'})$ in $T^*$.
           \item if $a$ is compared with ``null'' by $b$ (so $b$ is not fully matched in $N$), then we include $(a_i,b_k)$ in $T^*$ for some $k$ such that $(b_k,\ell(b_k)) \in N^*$ and $b_k$ is unmatched so far in $T^*$.
       \end{itemize}
       \item[--] suppose $b$ is compared with ``null'' by $a$ (so $a$ is not fully matched in $N$). 
       \begin{itemize}
         \item let $a'$ be the student that $b$ compares $a$ with in the evaluation of $\Delta_b(N,T)$ and so $(a'_{i'},b_{j'}) \in N^*$ for some $(i',j')$. We include the edge $(a_k,b_{j'})$ in $T^*$ for some $a_k$ such that $(a_k,\ell(a_k)) \in N^*$ and $a_k$ is unmatched so far in $T^*$. 
         \item in case $a$ is compared with ``null'' by $b$,
then we include the edge $(a_{k'},b_k)$ in $T^*$ for some $k'$ and $k$ such that  $(a_{k'},\ell(a_{k'}))$ and $(b_k,\ell(b_k))$ are in $N^*$ 
and $a_{k'}$ and $b_k$ are unmatched so far in $T^*$.
       \end{itemize}
      \end{itemize}
     \item[(iii)] For any vertex $u_k \in A' \cup B'$ that is left unmatched in steps~(i)-(ii), include $(u_k,\ell(u_k))$ in $T^*$.     
  \end{itemize}
  
  It is easy to see that $T^*$ is a valid matching in $G'_N$ and it matches all vertices in $A' \cup B'$.
  We have $\wt_N(T^*) = \sum_{e \in T^*}\wt_N(e)$.
  \begin{eqnarray}
    \sum_{e \in T^*}\wt_N(e) & = &  \sum_{(a_i,b_j) \in T^*} \big(\vote_a(T^*(a_i),N^*(a_i)) + \vote_b(T^*(b_j),N^*(b_j))\big) \nonumber\\
               &  & \ \ \ + \sum_{(u_k,\ell(u_k)) \in T^*} \vote_u(\ell(u_k),N^*(u_k)) \\
    & = & \ \ \sum_{u \in A\cup B}\sum_{i=1}^{\capac(u)}\vote_u(T^*(u_i),N^*(u_i)) \label{eq:popular1}\\
    & = & - \sum_{u \in A\cup B}\Delta_u(N,T) \label{eq:popular2}\\
    & = & -\Delta(N,T).\label{eq:popular3}
   \end{eqnarray}

  We have $\wt_N(a_i,b_j) = \vote_a(b,N^*(a_i)) + \vote_b(a,N^*(b_j))$ from the definition of edge weights  
  in $G'_N$. By grouping together for each vertex $u$, the edges $(u_i,v_j) \in T^*$ for all partners $v$ of $u$ in $T$ and any
  possible $(u_k,\ell(u_k))$ edges, we get the right side of Eqn.~\eqref{eq:popular1}.

  Crucially, Eqn.~\eqref{eq:popular2} follows from how we constructed the matching $T^*$: for each 
  vertex $u$, we have  $\sum_i\vote_u(N^*(u_i),T^*(u_i))  = \Delta_u(N,T)$ and so
  $\sum_i\vote_u(T^*(u_i),N^*(u_i)) = -\Delta_u(N,T)$.
  The total sum of all the terms $\Delta_u(N,T)$ for $u \in A' \cup B'$ is $\Delta(N,T)$.
  Thus it follows that $\wt_N(T^*) = -\Delta(N,T)$ and hence $N$ is a popular matching. \qed
\end{proof}

We now apply the above theorem to show that every pairwise-stable matching in $G$ is also a popular matching.
\begin{corollary}
  Every pairwise-stable matching in $G$ is popular.
\end{corollary}
\begin{proof}
  Let $S$ be any pairwise-stable matching in $G$. Consider the graph $G'_S$: since $S$ has no blocking edge in $G$, every edge $e$ in $G'_S$ satisfies $\wt_S(e) \le 0$. Thus every matching in $G'_S$ has weight at most 0 and so by
  Theorem~\ref{main-thm}, we can conclude that $S$ is popular. \qed
\end{proof}

Our goal is to show that $M_0$ is a max-size popular matching in $G$. We will do this as follows in Sections~\ref{sec:M-prop}
and \ref{sec:more-analysis}:
\begin{itemize}
\item We will show in Theorem~\ref{thm:max-pop} that $M_0$ satisfies the sufficient condition for popularity given in
  Theorem~\ref{main-thm}.
\item Lemma~\ref{lem:new-max-size} will show that no matching larger than $M_0$ can be a popular matching in $G$.
\end{itemize}
  
\subsection{The popularity of $M_0$}
\label{sec:M-prop}
We will now use Theorem~\ref{main-thm} to prove the popularity of the matching $M_0$ computed in 
Section~\ref{sec:hosp-res-algo}. We will construct the matchings $M'_0, M^*_0$ and the graph 
$G'_{M_0}$ corresponding to the matching $M_0$ as described at the beginning of
Section~\ref{sec:analysis}.
Our goal is to show that every $(A' \cup B')$-complete matching in $G'_{M_0}$ has weight at most 0.
Note that the matching $M^*_0$ has weight 0 in $G'_{M_0}$.

We partition the set $A'$ into $A'_0 \cup A'_1$ and the set $B'$ into $B'_0 \cup B'_1$ as follows.
Initialize $A'_0 = A'_1 = B'_0 = B'_1 = \emptyset$.
For each edge $(a_i,b_j) \in M'_0$ do:
\begin{itemize}
\item if $(a^0,b) \in M$ then add $a_i$ to $A'_0$ and $b_j$ to $B'_0$;
\item else (i.e., $(a^1,b) \in M$) add $a_i$ to $A'_1$ and $b_j$ to $B'_1$.
\end{itemize}
Recall that $M \subseteq A'' \times B$ is the matching in the graph $H$ obtained at the end of the 
2-level Gale-Shapley algorithm (see Algorithm~\ref{alg:many-many}) and the projection of $M$ on to $A \times B$ is $M_0$.

\begin{figure}[h]
	
  	\begin{center}
	
  	\begin{tikzpicture}
		
  		\def\width{1}
		\def\height{2.6*\width/1.62}
		
		\draw[dashed, thick] (-3*\width, 0) to (3*\width, 0);
		\draw[dashed, thick] (0,-1*\height) to (0, 1*\height);
		
		\node (r1) at (-2*\width, .5*\height) {$A'_1$};
		\node (r0) at (-2*\width, -.5*\height) {$A'_0$};
		\node (h1) at (2*\width, .5*\height) {$B'_1$};
		\node (h0) at (2*\width, -.5*\height) {$B'_0$};
		
		\draw (-2*\width,0) circle [thick, x radius=\width, y radius=\height];
		\draw (2*\width,0) circle [thick, x radius=\width, y radius=\height];
		
  	\end{tikzpicture}
	
  	\end{center}
	
	\caption{$A' = A'_0 \cup A'_1$ and $B' = B'_0 \cup B'_1$: all courses $b_j$ left unmatched in $M'_0$ are in $B'_0$ and all students $a_i$ left unmatched in $M'_0$ are in $A'_1$. Note that $M'_0 \subseteq (A'_0 \times B'_0) \cup (A'_1 \times B'_1)$.}
	\label{fig0:label}
	
  \end{figure}
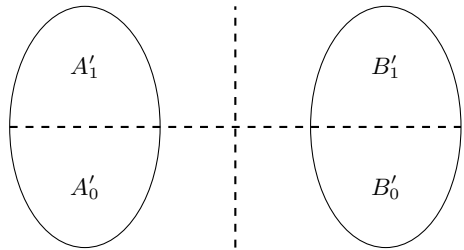

The definition of the sets $A'_0, A'_1, B'_0, B'_1$ implies that $M'_0 \subseteq (A'_0 \times B'_0) \cup (A'_1 \times B'_1)$.
Also add students unmatched in $M_0'$ to $A'_1$ and courses unmatched in $M_0'$ to $B'_0$. Thus 
we have $A' = A'_0 \cup A'_1$ and $B' = B'_0 \cup B'_1$ (see Fig.~\ref{fig0:label}).

Theorem~\ref{thm:max-pop} will show that the matching $M_0$ satisfies the condition of 
Theorem~\ref{main-thm}, this will prove that $M_0$ is a popular matching in $G$. 
This proof is inspired by the proof in \cite{Kav16} that shows the membership of 
certain half-integral matchings in the popular fractional matching polytope of a 
stable marriage instance.

In order to show that every $(A' \cup B')$-complete matching 
in $G'_{M_0}$ has weight at most 0, we consider the max-weight $(A' \cup B')$-complete matching 
problem in $G'_{M_0}$ as our primal LP. We show a dual feasible solution $\vec{\alpha}$ that makes
the dual objective function 0. This means the primal optimal value is at most 0 and this is what
we set out to prove.

\begin{theorem}
\label{thm:max-pop}
Every $(A' \cup B')$-complete matching in $G'_{M_0}$ has weight at most 0.
\end{theorem}
\begin{proof}
Let our primal LP be the max-weight $(A' \cup B')$-complete matching 
problem in $G'_{M_0}$. We want to show that the primal optimal value is at most 0.
The primal LP is the following: 
\begin{align*}
	\max \sum_{e \in G'_{M_0}} \wt_{M_0}(e)\cdot x_e&\\
	\text{subject to}\qquad\sum_{e \in E'(u_k)}x_e &= 1 &&\text{for all } u_k \in A' \cup B'\text,\\
	x_e &\ge 0 &&\text{for all edges } e \in G'_{M_0}\text,
\end{align*}
where $E'(u_k)$ is the set of edges incident on $u_k$ in $G'_{M_0}$.

The dual LP is the following: we associate a variable $\alpha_{u_k}$ to each vertex $u_k \in A' \cup B'$.
\begin{align}
	\min \sum_{u_k \in A' \cup B'}\alpha_{u_k}& \notag\\
	  \text{subject to}\qquad \alpha_{a_i} + \alpha_{b_j} &\ge \wt_{M_0}(a_i,b_j)  &&\text{for all edges}\ (a_i,b_j)\in G'_{M_0}\text,\label{eq:dual1}\\
	\alpha_{u_k} &\ge \wt_{M_0}(u_k,\ell(u_k)) &&\text{for all } u_k \in A' \cup B'\text.\label{eq:dual2}
\end{align}

Consider the following assignment of $\alpha$-values for all $u_k \in A' \cup B'$: set 
$\alpha_{u_k} = 0$ for all $u_k$ unmatched in $M'_0$ (each such vertex is in $A'_1\cup B'_0$) and for the 
matched vertices $u_k$ in $M'_0$, we set $\alpha$-values as follows: $\alpha_{u_k} = 1$ if $u_k \in A'_0\cup B'_1$
and $\alpha_{u_k} = -1$ if $u_k \in A'_1\cup B'_0$.

Observe that Inequality~\eqref{eq:dual2} holds for all vertices $u_k \in A' \cup B'$. This is because 
$\alpha_{u_k} = 0 = \wt_{M_0}(u_k,\ell(u_k))$ for all $u_k$ unmatched in $M'_0$; similarly, for all $u_k$ matched 
in $M'_0$ we have $\alpha_{u_k} \ge -1 = \wt_{M_0}(u_k,\ell(u_k))$.
In order to show Inequality~\eqref{eq:dual1}, we will use Claim~\ref{lem:G-M-edges} stated below.

\begin{new-claim}
\label{lem:G-M-edges}  
  Let $e = (a_i,b_j)$ be any edge in $G'_{M_0}$.
  \begin{itemize}
    \item[(i)] If $e \in A'_1 \times B'_0$, then $\wt_{M_0}(e) = -2$.
    \item[(ii)] If $e \in (A'_0 \times B'_0) \cup (A'_1 \times B'_1)$, then $\wt_{M_0}(e) \le 0$.
  \end{itemize}    
\end{new-claim}

\begin{itemize}
\item Claim~\ref{lem:G-M-edges}~(i) says that for every edge $(a_i,b_j) \in  A'_1 \times B'_0$ in $G'_{M_0}$, we have 
$\wt_{M_0}(a_i,b_j) = -2$.
Since  $\alpha_{u_k} \ge -1$ for all $u_k \in A'_1 \cup B'_0$, Inequality~\eqref{eq:dual1} holds for all 
edges of $G'_{M_0}$ in $A'_1 \times B'_0$. 
\item Claim~\ref{lem:G-M-edges}~(ii) says that for every edge $(a_i,b_j)$ in 
$(A'_0 \times B'_0) \cup (A'_1 \times B'_1)$, we have $\wt_{M_0}(a_i,b_j) \le 0$. Since 
$\alpha_{a_i} + \alpha_{b_j} \ge 0$ for all $(a_i,b_j) \in A'_t\times B'_t$ (for $t=0,1$),
Inequality~\eqref{eq:dual1} holds for all edges of $G'_{M_0}$ in  
$(A'_0 \times B'_0) \cup (A'_1 \times B'_1)$.
\end{itemize}

Since $\wt_{M_0}(e) \le 2$ for all edges $e$ in $G'_{M_0}$ and we set $\alpha_{u_k} = 1$ for all 
vertices $u_k \in A'_0\cup B'_1$, Inequality~\eqref{eq:dual1} is satisfied for all edges of $G'_{M_0}$ in  
$A'_0 \times B'_1$. Thus Inequality~\eqref{eq:dual1} holds for all edges $(a_i,b_j)$ in $G'_{M_0}$ and so these 
$\alpha$-values are dual feasible. 

For every edge $(a_i,b_j) \in M'_0$, we have $\alpha_{a_i} + \alpha_{b_j} = 0$ and $\alpha_{u_k} = 0$ for 
vertices $u_k$ unmatched in $M'_0$. Hence it follows that $\sum_{u_k \in A' \cup B'}\alpha_{u_k} = 0$. So by weak duality, 
the optimal value of the primal LP is at most 0. In other words, every matching in $G'_{M_0}$ that matches all vertices in 
$A' \cup B'$ has weight at most 0. \qed
\end{proof}

\noindent{\bf Proof of Claim~\ref{lem:G-M-edges}.}
  Consider any edge $(a_i,b_j) \in A'_1 \times B'_0$ in $G'_{M_0}$. Note that the matching $M_0$ does 
not contain the edge $(a,b)$ -- if it did, then $G'_{M_0}$ would have only one copy of this edge, 
say $(a_s,b_t)$, which being an edge of $M'_0$, has to be in either $A'_0 \times B'_0$ or 
$A'_1 \times B'_1$ whereas we are given that $(a_i,b_j) \in A'_1 \times B'_0$. 
The student $a_i \in A'_1$, i.e., $a^1$ got activated in our algorithm and recall 
that every course prefers level~1 neighbors to level~0 neighbors in our algorithm. So if $a^1$ 
had proposed to $b$, then this offer would have been accepted since $b$ had at least one partner 
who was a level~0 student (since $b_j \in B'_0$). Thus $a^1$ (with its entire residual capacity) 
must have been accepted by neighbors that $a$ prefers to $b$. Hence $a_i$ prefers its partner in 
$M^*_0$ to $b$, so $\vote_a(b,M^*_0(a_i)) = -1$.
  
  Since $(a,b) \notin M_0$ while $a^1$ got activated in our algorithm, along with the fact that $b_j \in B'_0$, 
it follows that the student $a^0$ was rejected by $b$. When $b$ rejected $a^0$, the course $b$ was matched 
to $\capac(b)$ neighbors, each of which was preferred by $b$ to $a^0$. Thereafter, $b$ may have received 
(and accepted) better offers from its neighbors and since $b_j \in B'_0$, the course $b$ never received enough 
offers from level~1 neighbors to have all its partners as level~1 students. In particular, $b_j$ is matched 
to a level~0 neighbor that is preferred to $a^0$. Thus $b_j$ prefers its neighbor in $M^*_0$ to $a$, so 
$\vote_b(a,M^*_0(b_j)) = -1$. So it follows that  $\wt_{M_0}(a_i,b_j) = -2$.

  We will now show part~(ii) of this lemma. In our algorithm, the preference order of each vertex,
when restricted to level~0 neighbors, is its original preference order and similarly, its 
preference order when restricted to level~1 neighbors, is its original preference order. Thus for 
each edge $(a_i,b_j)$ in $G'_{M_0}$ where $(a_i,b_j) \in (A'_0 \times B'_0) \cup (A'_1 \times B'_1)$, 
either (1)~the vertex $b$ prefers $M'_0(b_j)$ to $a$ or the vertex $a$ prefers $M'_0(a_i)$ to $b$  or 
(2)~$(a_i,b_j) \in M'_0$. In both cases, we have $\wt_{M_0}(a_i,b_j) \le 0$. \qed

\subsection{Maximality of the popular matching $M_0$}
\label{sec:more-analysis}

We need to show that $M_0$ is a max-size popular matching in $G$ and we now show that this follows quite easily
from the proof of Theorem~\ref{thm:max-pop}. Let $T$ be any matching in $G$. 
We can obtain a realization $T^*$ of the matching $T$ in $G'_{M_0}$ that is absolutely analogous to 
how it was done in the proof of Theorem~\ref{main-thm}. Thus $T^*$ is an
$(A' \cup B')$-complete matching in $G'_{M_0}$ and $\wt_{M_0}(T^*) = -\Delta(M_0,T)$.

We know from  Theorem~\ref{thm:max-pop} that $\wt_{M_0}(T^*) \le 0$. Suppose $T$ is a popular matching in $G$.
Then $\wt_{M_0}(T^*)$ has to be $0$, otherwise the popularity of $T$ is contradicted since $\wt_{M_0}(T^*) < 0$ implies that 
$\Delta(M_0,T) > 0$ (because $\wt_{M_0}(T^*) = -\Delta(M_0,T)$).

So if $T$ is a popular matching in $G$, then $T^*$ is an optimal solution to the maximum weight 
$(A'\cup B'$)-complete matching problem in $G'_{M_0}$. Recall that this is the primal LP in the proof 
of Theorem~\ref{thm:max-pop}. We will use the dual feasible solution $\vec{\alpha}$ that we constructed in
the proof of Theorem~\ref{thm:max-pop} and apply complementary slackness to show that if $(u_k,\ell(u_k)) \in M_0^*$,
i.e., if $u_k$ is left unmatched in $M_0'$, then  $T^*$ also has to contain the edge $(u_k,\ell(u_k))$. This
will imply that $|T| \le |M_0|$, i.e., every popular matching in $G$ has size at most $|M_0|$.

\begin{lemma}
  \label{lem:new-max-size}
  Let $T$ be a popular matching in $G$ and let $T^*$ be the realization of $T$ in $G'_{M_0}$.
  Then for any vertex $u_k \in A' \cup B'$ we have: $(u_k,\ell(u_k)) \in M_0^*$ implies $(u_k,\ell(u_k)) \in T^*$.
\end{lemma}
\begin{proof}
Consider the $\alpha$-values assigned to vertices in $A' \cup B'$ in the proof of 
Theorem~\ref{thm:max-pop}. This is an {\em optimal dual} solution since its value is 0 
which is the value of the optimal primal solution. Thus complementary slackness conditions
have to hold for each edge in the
optimal solution $(T^*_e)_{e \in G'_{M_0}}$ to the primal LP. That is, for each edge $(u_k,v_t) \in G'_{M_0}$, 
we have:
\begin{equation}
  \label{eqn:new}
  \mathrm{either}\ \ \ \  \alpha_{u_k} + \alpha_{v_t} = \wt_{M_0}(u_k,v_t) \ \ \ \ \mathrm{or}\ \ \ \ T^*_{(u_k,v_t)} = 0.
\end{equation}  

Let $u_k \in A' \cup B'$ be a vertex such that $(u_k,\ell(u_k)) \in M_0^*$, so $\alpha_{u_k} = 0$. If 
$u \in A$, then $u_k \in A'_1$. Observe that all of $u_k$'s neighbors in $G'_{M_0}$ are in $B'_1$  -- this is
because for any neighbor $v_t \ne \ell(u_k)$ of $u_k$, we have $\vote_u(v,\ell(u_k)) = 1$ and so $\wt_{M_0}(u_k,v_t) \ge 0$.
Claim~\ref{lem:G-M-edges}~(i) says that $\wt_{M_0}(u_k,v_t) = -2$ for all edges
$(u_k,v_t) \in A'_1 \times B'_0$. Thus $u_k$ has no neighbor in $B'_0$.
Similarly, if $u \in B$, then $u_k \in B'_0$ and all its neighbors in $G'_{M_0}$ are in $A'_0$; otherwise 
$u_k$ has a neighbor $v_t$ in $A'_1$ and Claim~\ref{lem:G-M-edges}~(i) would get contradicted since $\wt_{M_0}(u_k,v_t) \ge 0$.

In both cases, every edge $(u_k,v_t) \in A' \times B'$ that is incident on $u_k$ in $G'_{M_0}$ is 
{\em slack} because $(u_k,v_t) \in (A'_0\times B'_0) \cup (A'_1 \times B'_1)$: thus 
$\alpha_{u_k} = 0$ and $\alpha_{v_t} = 1$ while $\wt_{M_0}(u_k,v_t)=\vote_u(v,\ell(u_k))+\vote_v(u,M'_0(v_t))=1-1=0$. 
Thus it follows from Equation~(\ref{eqn:new}) that  $T^*_{(u_k,v_t)} = 0$ for
$v_t \ne \ell(u_k)$. Since $T^*$ is $(A' \cup B')$-complete, we have $(u_k,\ell(u_k)) \in T^*$. \qed
\end{proof}

Now it is immediate to see that $M_0$ is a max-size popular matching in $G$. Let $T$ be any popular 
matching in $G$. Consider the matching $T' = T^* \setminus \{(u_k,\ell(u_k)): u_k \in A' \cup B'\}$. 
Lemma~\ref{lem:new-max-size} implies that $|T'| \le |M'_0|$ because every vertex $u_k$ 
left unmatched in $M'_0$ has to be left unmatched in $T'$ also. Since $|T| = |T'|$ and $|M'_0| = |M_0|$, we have 
$|T| \le |M_0|$.
As this holds for any popular matching $T$ in $G$, we can conclude that $M_0$ is a max-size popular 
matching in $G$.

It is easy  to see that our algorithm runs in linear time (see Appendix~B for the details).
Hence we can conclude the following theorem.

\begin{theorem}\label{thm:max-size}
  A max-size popular matching in a many-to-many instance $G = (A \cup B,E)$ can be computed in linear time.
\end{theorem}

\noindent{\bf A max-size weakly popular matching.}
We will now show that no matching larger than $M_0$ can be weakly popular
(see Definition~\ref{def:weakly-pop}) as $\Delta(M_0,T) > 0$ for any such matching $T$. This implies that $M_0$
is also a max-size weakly popular matching in $G$.

\begin{lemma}
\label{lem:max-size}
  Let $T$ be a matching such that $|T| > |M_0|$. Then $\Delta(M_0,T) > 0$, i.e., $M_0$ is more 
popular than $T$.
\end{lemma}
\begin{proof}
Let $T$ be a larger matching than $M_0$. So some vertices left unmatched in $M_0$ have to be matched
in $T$. Thus for some $u_k \in A' \cup B'$ such that $(u_k,\ell(u_k)) \in M^*_0$, the matching
$T^*$ (which is the $(A'\cup B'$)-complete matching in $G'_{M_0}$ corresponding to $T$)
has to contain an edge $(u_k,v_t)$ where $v_t \ne \ell(u_k)$.
Thus $T^*$ contains a {\em slack} edge $(u_k,v_t) \in A' \times B'$ since $\alpha_{u_k} = 0$, 
$\alpha_{v_t} = 1$ while $\wt_{M_0}(u_k,v_t) = 0$ (see the proof of Lemma~\ref{lem:new-max-size}).

It now follows from Equation~(\ref{eqn:new}) that $T^*$ cannot be an optimal solution to the 
maximum weight $(A'\cup B'$)-complete matching problem in $G'_{M_0}$. Thus $\wt_{M_0}(T^*) < 0$, 
in other words, $\Delta(M_0,T) > 0$ since $\wt_{M_0}(T^*) = -\Delta(M_0,T)$. \qed
\end{proof}

Interestingly, Lemma~\ref{lem:max-size} implies that for any definition of popularity that is ``in between''
popularity and weak popularity, the size of a max-size popular matching is the same.
To formalize the meaning of ``in between'',
consider the two relations on matchings $\succsim_{\mathit{p}}$ and $\succsim_{\mathit{wp}}$, where $M_0\succsim_{\mathit{p}} M_1$ if
$\Delta(M_0,M_1)\ge 0$ and $M_0\succsim_{\mathit{wp}}~M_1$ if $\Delta(M_1,M_0)\le 0$, induced by popularity and weak popularity,
respectively. Clearly, ${\succsim_{\mathit{p}}}~\subseteq~{\succsim_{\mathit{wp}}}$. Note that  popular matchings and weakly popular
matchings correspond to
maximal elements of $\succsim_{\mathit{p}}$ and $\succsim_{\mathit{wp}}$, respectively.\footnote{$M_0$ is a maximal element of a relation $\succsim$ if for all elements $M_1$ we have: $M_1\succsim M_0$ implies $M_0\sim M_1$.}
We showed that $M_0$, which is a max-size maximal element of $\succsim_{\mathit{p}}$, is also a max-size maximal element of
$\succsim_{\mathit{wp}}$. This implies that if $\succsim$ is a relation on matchings (induced by an alternative notion of popularity)
such that ${\succsim_{\mathit{p}}} \subseteq {\succsim} \subseteq {\succsim_{\mathit{wp}}}$, then $M_0$ is also a max-size maximal
element of $\succsim$. This allows us to conclude the following proposition which even allows for different vertices
to compare sets of neighbors in different ways.

\begin{proposition}
\label{prop1}
  The size of a max-size popular matching in $G = (A \cup B,E)$ is invariant to the way 
vertices compare sets of neighbors as long as it is in between the most adversarial and the most favorable comparison.
\end{proposition}

\subsection{The rural hospitals theorem for max-size popular matchings}
The rural hospitals theorem for stable matchings~\cite{Roth86a} does not necessarily hold for
max-size popular matchings. That is, a hospital that is not matched up to 
capacity in some max-size popular matching is not necessarily matched to the same set of residents in every max-size 
popular matching.

Consider the instance $G = (R\cup H, E)$ with $R = \{r,r'\}$ and
$H = \{h,h'\}$ and $\capac(h) = 1$ and $\capac(h') = 2$.
The edge set is $R\times H$. 
The preferences are shown in the table below.
The (max-size) popular matchings are $M = \{(r,h),(r',h')\}$ (in black) and $M' = \{(r,h'),(r',h)\}$ (in red).
So $h'$ is matched to a different resident in the two max-size popular matchings $M$ and $M'$.
Note that $M'$ is not stable, as $(r,h)$ is a blocking pair.

	\begin{center}
		
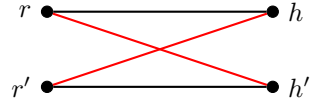
		\begin{minipage}[c]{0.45\textwidth}
			
			\centering
			\begin{align*}
				&r\colon h, h'\qquad\qquad &&h\colon r, r'\\
				&r'\colon h, h' &&h'\colon r, r'\\
			\end{align*}
			
		\end{minipage}
		\hfill
		\begin{minipage}[c]{0.45\textwidth}
			
			\centering 
  			\begin{tikzpicture}
		
  				\def\radius{.15cm}
		
  				\node[label = left:{$r$}, draw, circle, fill, minimum size = \radius, inner sep = 0pt] (r) at (0,.5) {};
  				\node[label = left:{$r'$}, draw, circle, fill, minimum size = \radius, inner sep = 0pt] (r') at (0,-.5) {};

  				\node[label = right:{$h$}, draw, circle, fill, minimum size = \radius, inner sep = 0pt] (h) at (3,.5) {};
  				\node[label = right:{$h'$}, draw, circle, fill, minimum size = \radius, inner sep = 0pt] (h') at (3,-.5) {};

  				\draw[black, thick] (r) to (h);
  				\draw[black, thick] (r') to (h');
  				\draw[red, thick] (r') to (h);		
  				\draw[red, thick] (r) to (h');		
					
  			\end{tikzpicture}
	
		\end{minipage}
	
	\end{center}

However Lemma~\ref{lem:rural-hosp} holds here. Such a result for max-size popular matchings in the 
one-to-one setting (that every max-size popular matching has to match the same set of vertices) 
was shown in \cite{Hirakawa-MatchUp15}. Our proof is based on linear programming and is different
from the combinatorial proof in \cite{Hirakawa-MatchUp15}.

\begin{lemma}
\label{lem:rural-hosp}
Let $T$ be a max-size popular matching in $G$. Then $T$ matches the same vertices as $M_0$ (the matching computed in
Section~\ref{sec:hosp-res-algo}) and 
moreover, every vertex $u$ is matched in $T$ to the same capacity as it gets matched to in $M_0$.
\end{lemma}
\begin{proof}
Consider the realization $T^*$ of $T$ such that
$T^*$ is an $(A' \cup B')$-complete matching in $G'_{M_0}$ and $\wt_{M_0}(T^*) = -\Delta(M_0,T)$.
Since $T$ is popular, we know that $T^*$ has to include all the edges $(u_k,\ell(u_k))$ for vertices $u_k$
left unmatched in $M'_0$ (by Lemma~\ref{lem:new-max-size}).
Let $T' = T^* \setminus \{(u_k,\ell(u_k)): u_k \in A' \cup B'\}$.
Every vertex in $G'_{M_0}$ that is unmatched in $M'_0$ is left unmatched in $T'$ also.
We also have $|M'_0| = |M_0| = |T| = |T'|$ as both $M_0$ and $T$ are max-size popular matchings in $G$.
Hence $T'$ and $M'_0$ match the same vertices in $G'_{M_0}$, i.e., every 
vertex $v$ in $G$ is matched in $T$ to the same capacity as it gets matched to in $M_0$. \qed
\end{proof}

We also show the following results here: Lemma~\ref{lem:size-S} bounds the size of $M_0$, where
$M_0$ is a max-size popular matching in $G$ and Lemma~\ref{lem:min-size} shows
that a pairwise-stable matching is a min-size popular matching in $G$.
The proofs of Lemmas~\ref{lem:size-S} and \ref{lem:min-size} are inspired by
analogous proofs in the one-to-one setting shown in \cite{Kav12} and in \cite{HK11}, respectively.

\begin{lemma}
  \label{lem:size-S}
  $|M_0| \ge \frac{2}{3}|M_{\max}|$, where $M_{max}$ is a max-size matching in $G$.
\end{lemma}
\begin{proof}
  The size of the matching $M_0$ is exactly the same as $M'_0$. Consider the graph $G'_{M_0}$ without
 last resort neighbors (call this graph $G''_{M_0}$) -- we will show that every augmenting path with 
 respect to $M_0'$ here has length at least 5. This will immediately imply that $|M'_0| \ge 2c/3$ 
 where $c$ is the size of the max-size matching in $G''_{M_0}$. The value $c \ge |M_{\max}|$ since 
 corresponding to any matching $T$ in $G$, we have a matching $T'$ in  $G''_{M_0}$ such that 
 $|T| = |T'|$. Thus the size of a max-size matching in $G''_{M_0}$ is at least $|M_{\max}|$ and so we get
 $|M_0| \ge 2|M_{\max}|/3$.

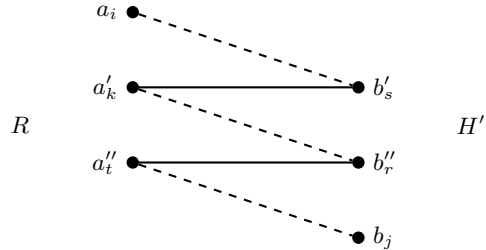
\begin{figure}[h]
	
  	\begin{center}
	
  	\begin{tikzpicture}
		
  		\def\radius{.15cm}
		
  		\node[label = left:{$a_i$}, draw, circle, fill, minimum size = \radius, inner sep = 0pt] (r_1) at (0,0) {};
  		\node[label = left:{$a'_k$}, draw, circle, fill, minimum size = \radius, inner sep = 0pt] (r_2) at (0,-1) {};
  		\node[label = left:{$a''_t$}, draw, circle, fill, minimum size = \radius, inner sep = 0pt] (r_3) at (0,-2) {};
  	
  		\node[label = right:{$b'_s$}, draw, circle, fill, minimum size = \radius, inner sep = 0pt] (h_2) at (3,-1) {};
  		\node[label = right:{$b''_r$}, draw, circle, fill, minimum size = \radius, inner sep = 0pt] (h_3) at (3,-2) {};
  		\node[label = right:{$b_j$}, draw, circle, fill, minimum size = \radius, inner sep = 0pt] (h_4) at (3,-3) {};
		
  		\draw[thick] (r_2) to (h_2);
  		\draw[thick] (r_3) to (h_3);

  		\draw[dashed, thick] (r_1) to (h_2);
  		\draw[dashed, thick] (r_2) to (h_3);
  		\draw[dashed, thick] (r_3) to (h_4);
		
		\node (R) at (-1.5, -1.5) {$R$};
		\node (H') at (4.5, -1.5) {$H'$};
		
  	\end{tikzpicture}
	
  	\end{center}
	
    \caption{An augmenting path with respect to $M'_0$ in the graph $G''_{M_0}$: the vertices $a_i, a'_k$ have to belong to $A'_1$ while the vertices  $b_j, b''_r$ have to belong to $H'_0$. Thus the length is of this path is $\ge 5$.}
    \label{fig3:label}
	
  \end{figure}

 Consider any augmenting path $p$ with respect to $M'_0$ in the graph $G''_{M_0}$  (see Fig.~\ref{fig3:label}) --
 let the endpoints of $p$ be $a_i$ and $b_j$. Since these vertices are left unmatched in $M'_0$, it follows that
 $a_i \in A'_1$ and $b_j \in B'_0$. As seen in the proof of Lemma~\ref{lem:new-max-size}, the vertex $a_i$ is adjacent
 only to vertices in $B'_1$ in the graph $G''_{M_0}$ and the unmatched vertex $b_j$ is adjacent only to vertices in $A'_0$
 in the graph $G''_{M_0}$. Every vertex in $A'_0$ is matched in $M_0$ to a neighbor in $B'_0$ and every vertex in $B'_1$ is
 matched in $M_0$ to a neighbor in $A'_1$. Thus the shortest augmenting path with respect to $M'_0$ has the following structure
 with respect to the sets in $\{A'_0,A'_1,B'_0,B'_1\}$ that its vertices belong to: $A'_1$-$B'_1$-$A'_1$-$B'_0$-$A'_0$-$B'_0$, i.e.,
 its length is at least 5. \qed
\end{proof}

\begin{lemma}
\label{lem:min-size}  
  A pairwise-stable matching $S$ is a min-size weakly popular matching in $G$.
\end{lemma}
\begin{proof}
  Let $T$ be a matching in $G$ such that $|T| < |S|$. Consider a realization $T^*$ of the matching $T$ in the graph $G'_S$ as
  described in the proof of Theorem~\ref{main-thm} such that $T^*$ is $(A' \cup B')$-complete and $\wt_S(T^*) = -\Delta(S,T)$.
  Recall that $S$ is a pairwise-stable matching in $G$ -- hence for each edge $e$
  in $G'_S$, we have $\wt_S(e) \le 0$. Moreover, because $|T| < |S|$, there is a
  vertex $u_i$ that is matched to a genuine neighbor in $S$, however $T^*$ contains the edge $e = (u_i,\ell(u_i))$.
  We have $\wt_S(e) = -1$. Thus $\wt_S(T^*) < 0$. In other words, $\Delta(S,T) > 0$ and $T$ cannot be weakly popular. Since $S$ 
  is a popular matching in $G$, it means that $S$ is a min-size popular matching in $G$, in fact, $S$ is a min-size weakly popular
  matching in $G$.  \qed
\end{proof}

\medskip

\noindent{\em Acknowledgments.} The first author wishes to thank Larry Samuelson for comments on the motivation for popular matchings.
The second author wishes to thank David Manlove and Bruno Escoffier for asking her about popular matchings in the hospitals/residents setting.

\section*{Appendix}

\setcounter{subsection}{0}
\renewcommand{\thesubsection}{\Alph{subsection}}

\subsection{A naive approach for finding max-size popular matchings}

Given a many-to-many matching instance $G = (A \cup B,E)$, we investigate the possibility of constructing a corresponding one-to-one
matching instance $G' = (A'\cup B',E')$ (with \emph{strict} preference lists) in order to show a reduction from the max-size popular matching 
problem in $G$ to one in $G'$. The vertex set $A'$ will have $\capac(a)$ many copies $a_1, a_2, \ldots$ of every $a \in A$ and 
$B'$ will have $\capac(b)$ many copies $b_1, b_2, \ldots$ of every $b \in B$; the edge set $E'$ has $\capac(a)\cdot\capac(b)$ many 
copies of edge $(a,b)$ in $E$. If $v \succ_u v'$ in $G$ then we have $v_i \succ_{u_k} v'_j$ for each
$i \in \{1,\ldots,\capac(v)\}$, $j \in \{1,\ldots,\capac(v')\}$, and $k \in \{1,\ldots,\capac(u)\}$. Among the copies 
$v_1,\ldots,v_{\capac(v)}$ of the same vertex $v$, we will set $v_1 \succ_{u_k} \cdots \succ_{u_k} v_{\capac(v)}$, for each vertex $u_k$ in $G'$.

Given any matching $\tilde{M}$ in $G'$, we define $\proj(\tilde{M})$ as the projection of $\tilde{M}$, which is the matching obtained by 
dropping the subscripts of all vertices; $\proj(\tilde{M})$ obeys all capacity bounds in $G$.
We will now consider the  {\em many-to-one} or the hospitals/residents setting: so there
are no multi-edges in $\proj(\tilde{M})$. It would be interesting to be able to
show that every popular matching $M$ in $G$ has a {\em realization} $\tilde{M}$ 
in $G'$ (i.e., $\proj(\tilde{M}) = M$) such that $\tilde{M}$ is a popular matching in $G'$.

However the above statement is not true as shown by the following example. Let
$G = (R \cup H, E)$ where $R = \{p,q,r,s\}$ and $H = \{h,h',h''\}$ where $\capac(h) = 2$ and
$\capac(u) = 1$ for all other vertices $u$. The preference lists are as follows:

\begin{minipage}[c]{0.45\textwidth}
			
			\centering
			\begin{align*}
				&p\colon h, h'' \qquad\qquad &&h\colon p, q, r, s\\
				&q\colon h, h' &&h'\colon q\\
                                &r\colon h      &&h''\colon p\\
                                &s\colon h      &&\\
			\end{align*}
		\end{minipage}

Consider the matching $N = \{(p,h),(q,h'),(r,h)\}$. We show below that $N$ satisfies the sufficient
condition for popularity as given in Theorem~\ref{main-thm}. The proof of Claim~\ref{Appendix-clm1}
follows the same approach as used in the proof of Theorem~\ref{thm:max-pop}.

\begin{new-claim}
\label{Appendix-clm1}
$N$ is popular in $G$.
\end{new-claim}
\begin{proof}
We have $R' = \{p_1,q_1,r_1,s_1\}$ and $H' = \{h_1,h_2,h'_1,h''_1\}$. Here we use the notation introduced at the beginning of 
Section~\ref{sec:analysis}: let $N' = \{(p_1,h_1),(q_1,h'_1),(r_1,h_2)\}$ (see Fig.~\ref{fig:Appendix-example}). 

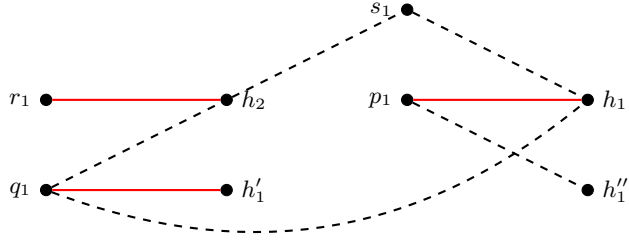
\begin{figure}[h]
  	\begin{center}
	
  	\begin{tikzpicture}
		
		\def\size{1.2}
  		\def\radius{.15cm}
		
  		\node[label = left:{$r_1$}, draw, circle, fill, minimum size = \radius, inner sep = 0pt] (r_1) at (0,0) {};
  		\node[label = left:{$q_1$}, draw, circle, fill, minimum size = \radius, inner sep = 0pt] (q_1) at (0,-1*\size) {};
  		\node[label = right:{$h_2$}, draw, circle, fill, minimum size = \radius, inner sep = 0pt] (h_2) at (2*\size,0) {};
  		\node[label = right:{$h_1'$}, draw, circle, fill, minimum size = \radius, inner sep = 0pt] (h_1') at (2*\size,-1*\size) {};
		
  		\node[label = left:{$s_1$}, draw, circle, fill, minimum size = \radius, inner sep = 0pt] (s_1) at (4*\size,1*\size) {};
  		\node[label = left:{$p_1$}, draw, circle, fill, minimum size = \radius, inner sep = 0pt] (p_1) at (4*\size,0) {};
  		\node[label = right:{$h_1$}, draw, circle, fill, minimum size = \radius, inner sep = 0pt] (h_1) at (6*\size,0) {};
  		\node[label = right:{$h_1''$}, draw, circle, fill, minimum size = \radius, inner sep = 0pt] (h_1'') at (6*\size,-1*\size) {};
		
  		\draw[red, thick] (r_1) to (h_2);
		\draw[red, thick] (q_1) to (h_1');
		\draw[red, thick] (p_1) to (h_1);
		
  		\draw[dashed, thick] (q_1) to (h_2);
  		\draw[dashed, thick] (h_2) to (s_1);
  		\draw[dashed, thick] (s_1) to (h_1);
		\draw[dashed, thick] (p_1) to (h_1'');
		\draw[dashed, thick, bend right] (q_1) to (h_1);
		
  	\end{tikzpicture}
  	\end{center}
\caption{The edges of the matching $N'$ are in red and the non-matching edges in $G'_N$ are dashed. For simplicity, we have not included last resort neighbors here. Note that the edge $(q_1,h_2)$ is a blocking edge to $N'$ as both $q_1$ and $h_2$ prefer each other to their respective partners in $N'$, i.e., $\wt_{N}(q_1,h_2) = 2$ and $\wt_{N}(e) = 0$ for all
  other edges $e$ in $G'_N$.}
\label{fig:Appendix-example}
\end{figure}

We need to show that every $(R'\cup H')$-complete matching in 
the weighted graph $G'_N$ has weight at most 0. 
We will show this by constructing a {\em witness} or a solution to the dual LP corresponding to the primal LP which is the 
$(R'\cup H')$-complete max-weight matching problem in $G'_N$. 
This solution is the following:
$\alpha_{p_1} = \alpha_{h_1} = \alpha_{s_1} = \alpha_{h''_1} = 0$ while $\alpha_{q_1} = \alpha_{h_2} = 1$ and $\alpha_{r_1} = \alpha_{h'_1} = -1$.
The above solution is dual-feasible since every edge in $G'_N$ is covered by the sum of $\alpha$-values of its endpoints -- in particular, note that
$\alpha_{q_1} + \alpha_{h_2} = 2 = \wt_{N}(q_1,h_2)$. The dual
optimal solution is at most $\sum_{u \in R' \cup H'}\alpha_u = 0$.
So the primal optimal solution is also at most 0, in other words, $N$ is a popular matching in $G$. \qed
\end{proof}

Note that the graph $G'$ has two {\em extra} edges relative to $G'_N$: these are $(p_1,h_2)$
and $(r_1,h_1)$. With respect to realizations of $N$ in $G'$, there are 2 candidates: these are
$N_1 = \{(p_1,h_1),(q_1,h'_1),(r_1,h_2)\}$ and $N_2 = \{(p_1,h_2),(q_1,h'_1),(r_1,h_1)\}$. 

\begin{new-claim}
Neither $N_1$ nor $N_2$ is popular in $G'$.
\end{new-claim}
\begin{proof}
Consider the matching $M_1 = \{(p_1,h''_1),(q_1,h_2),(r_1,h_1)\}$. The vertices $p_1$, $h_1$, and $h'_1$ prefer $N_1$ to $M_1$ while the
vertices $q_1$, $h_2$, $r_1$, and $h''_1$ prefer $M_1$ to $N_1$ and $s_1$ is indifferent.
Thus $M_1$ is more popular than $N_1$, i.e., $N_1$ is not a popular matching in $G'$.

Consider the matching $M_2 = \{(p_1,h_1),(q_1,h'_1),(s_1,h_2)\}$. The vertices $r_1$ and $h_2$ prefer $N_2$ to $M_2$ while the vertices
$p_1$, $h_1$, and $s_1$ prefer $M_2$ to $N_2$ and $q_1$, $h'_1$, and $h''_1$ are indifferent.
Thus $M_2$ is more popular than $N_2$, i.e., $N_2$ is not a popular matching in $G'$. \qed
\end{proof}

Note that the above instance can easily be transformed to another instance $G$ with a {\em max-size} popular matching that cannot
be realized as a popular matching in $G'$. 
In fact, it is easy to show that every popular matching in $G'$ gets projected to a popular 
matching in $G$. However as illustrated by the example above, there may exist popular matchings in $G$ that cannot be realized as popular
matchings in $G'$.

\subsection{Running time of our algorithm}
It is known that the Gale-Shapley algorithm for the hospitals/residents problem (and thus in the many-to-many setting)
can be implemented to run in linear time~\cite{manlove2015}. Our algorithm is very similar to the Gale-Shapley algorithm
in the many-to-many setting. For the sake of completeness, we describe a simple linear time implementation of our algorithm.

Recall that our algorithm runs in the graph $H$ whose vertex set is $A'' \cup B$. For each vertex $u \in A'' \cup B$, right at
the beginning of the algorithm,
we will construct an array $D_u$ that contains all neighbors of $u$ in decreasing order of preference.
More precisely, $D_u$ is an array of length $\deg_H(u)$ where $D_u[i]$ contains the identity of the $i$-th ranked neighbor of $u$.  
For $a \in A''$ and $1 \le i \le \deg_H(a)$, the cell $D_a[i]$ also stores the value $\rank_b(a)$, where $b$ is $a$'s $i$-th
ranked neighbor in $H$ and $\rank_b(a)$ is the rank of $a$ in $b$'s preference order in $H$.

It is easy to see that all the arrays $D_u$ can be constructed in  $O(m+n)$ time, where $m = |E|$ and $n = |A|+|B|$.
We will maintain the function $\maxrank(\cdot)$ defined below for each $b \in B$:
\begin{equation*} 
\label{maxrank-defn} 
    \maxrank(b) = \begin{cases} \text{rank of $b$'s worst neighbor in $M$} & \text{if $b$ is fully matched in $M$}\\
                                         \infty & \text{otherwise}.
    \end{cases}
\end{equation*}    

For any $a \in A''$, before $a$ proposes to a neighbor $b$, the vertex $a$ first checks if $\rank_b(a) < \maxrank(b)$ and goes ahead
with the proposal only if this condition is satisfied. If this condition is not satisfied, i.e., if $b$ ranks $a$ worse than
$\maxrank(b)$, then the edge $(a,b)$ is actually missing in the current graph $H$ and so $a$ does not propose to $b$.

For any $b \in B$, when $b$ receives $a$'s proposal (we assume that $a$'s proposal also contains the value $j = \rank_b(a)$),
$b$ accepts this proposal by setting a flag in the cell $D_b[j]$ to {\em true}. Thus for each $b \in B$ and
$1 \le j \le \deg_H(b)$, along with the identity $a$ of $b$'s $j$-th ranked neighbor in $H$, the cell $D_b[j]$
also contains a flag set to {\em true} if $(a,b) \in M$, else the flag is set to {\em false}. Thus the entire
matching $M$ is stored via these flags in the arrays $D_b$ for $b \in B$.

If $|M(b)| > \capac(b)$ after $b$ accepts $a$'s proposal, then we set the flag in the cell $D_b[\maxrank(b)]$ to
{\em false}. Also, $\maxrank(b)$ needs to be updated and this is done by using a pointer to traverse the array $D_b$
leftward from its current location to find the rightmost cell whose flag is set to {\em true}.

Consider the step where we check if $b$ is already matched to $u^0$ when $u^1$ proposes to $b$.
This check can be easily performed by comparing $u^0$'s rank in $b$'s preference order with $\maxrank(b)$. Note that $u^0$'s rank
in $b$'s preference order is $\rank_b(u^1) + \deg_G(b)$. If $(u^0,b) \in M$, then we need to replace $(u^0,b)$ in $M$ with $(u^1,b)$ --
this step is implemented by setting the flag in $D_b[\rank_b(u^0)]$ to {\em false} and the flag in $D_b[\rank_b(u^1)]$ to {\em true},
so this takes unit time.
It is easy to see that our entire algorithm can be implemented to run in linear time.
\end{document}